%% file: main.tex
\documentclass[12pt]{article}
\usepackage{times}
\usepackage[margin=1.5in]{geometry}
\usepackage{setspace}
\usepackage{fullpage}
\usepackage{amsmath, amsthm, amssymb}
\usepackage{verbatim}
\usepackage[T1]{fontenc}
\usepackage[varqu]{zi4}
\usepackage[libertine]{newtxmath}
\usepackage[sort]{natbib}
\usepackage{xcolor}
\usepackage[pagebackref=true,colorlinks]{hyperref}
\hypersetup{linkcolor=blue,filecolor=blue,citecolor=blue,urlcolor=blue}
\usepackage{graphicx}

\newtheorem{theorem}{Theorem}[section]

\newtheorem{proposition}[theorem]{Proposition}
\newtheorem{lemma}[theorem]{Lemma}

\theoremstyle{definition}
\newtheorem{definition}[theorem]{Definition}
\newtheorem{example}[theorem]{Example}

\input{www/commands}
\newcommand{\modelname}{Google Bard}

\newcommand{\rollbackfootnotecounter}{\rollbackfootnotecounterreal}

\newcommand{\tableonesize}{\scriptsize}

\renewcommand{\hybridtext}[2]{\longversion{#1}{#2}}

\title{Mechanism Design for Large Language Models\thanks{An extended abstract appeared in the Proceedings of WWW 2024. The work of Haifeng Xu was done as Visiting Faculty at Google Research.
We thank Dirk Bergemann, Marina Halac, Philipp Strack, Elliot Lipnowski, Yang Cai, Vasilis Syrgkanis, Negin Gorezaei,  Ido Cohen, Yoav Nagel, Yael Shemesh as well as the participants of the Yale Economics Seminar, the Stanford MS\&E Seminar, and the WWW 2024 conference for invaluable comments and suggestions to improve this manuscript. We are specially grateful to Yong Cheng from Google DeepMind for his expert guidance on the LLM-related details and literature.}}
\date{}

\author{Paul D\"{u}tting\thanks{Google Research, \texttt{\{duetting,mirrokni,renatoppl,szuo\}@google.com}.} \and Vahab Mirrokni$^{\dag}$ \and Renato Paes Leme$^{\dag}$ \and Haifeng Xu\thanks{University of Chicago \& Google Research, \texttt{haifengxu@uchicago.edu}.} \and Song Zuo$^{\dag}$}

\begin{document}

\maketitle

\begin{abstract}
  \input{abstract}
\end{abstract}

\input{body}

\bibliographystyle{ACM-Reference-Format}
\bibliography{www/sample-base}

\appendix
\input{www/appendix}

\end{document}

%% file: www/commands.tex
\usepackage{tikz}
\usepackage{makecell}

\usetikzlibrary{calc}

\usepackage{tablefootnote}
\usepackage{nicefrac}
\usepackage{xspace}

\newtheorem*{theorem*}{Theorem}

\newcounter{mytablefootnote}
\newcommand{\tablefootnotemark}{\protect\footnotemark{}\stepcounter{mytablefootnote}}

\newcommand{\rollbackfootnotecounterreal}{\addtocounter{footnote}{-\value{mytablefootnote}}}

\newcommand{\tablefootnotemarknullcite}[1]{\tablefootnotemark}

\newcommand{\tablefootnotetextnullcite}[1]{}

\newcommand{\Infer}{\Longrightarrow}

\newcommand{\M}{\mathcal{M}}

\newcommand{\R}{\mathbb{R}}
\newcommand{\Bids}{\mathbb{Q}}

\renewcommand{\b}{\boldsymbol{b}}
\newcommand{\p}{\boldsymbol{p}}

\newcommand{\Sampler}{\textsf{Sampler}}
\renewcommand{\L}{\mathcal{L}}
\newcommand{\KL}{\textsf{KL}}
\newcommand{\RL}{\textsf{RL}}

\newcommand{\Loss}{\textsc{Loss}}
\newcommand{\alf}{aggregated loss function\xspace}
\newcommand{\alfs}{aggregated loss functions\xspace}
\newcommand{\analf}{an \alf}

\newcommand{\Tunder}{T_+}
\newcommand{\Tover}{T_-}

\DeclareMathOperator{\E}{\mathbb{E}}

\newcommand{\longversion}[2]{#2}

\newcommand{\hybridtext}[2]{\shortversion{#1}{#2}}

%% file: abstract.tex
We investigate auction mechanisms for AI-generated content, focusing on applications like ad creative generation. In our model, agents' preferences over stochastically generated content are encoded as large language models (LLMs). We propose an auction format that operates on a token-by-token basis, and allows LLM agents to influence content creation through single dimensional bids. We formulate two desirable incentive properties and prove their equivalence to a monotonicity condition on output aggregation. This equivalence enables a second-price rule design, even absent explicit agent valuation functions. Our design is supported by demonstrations on a publicly available LLM.

%% file: body.tex
\section{Introduction}
In the current web ecosystem, auctions are the primary mechanism used to decide which ads (and commercial content more broadly) are displayed to users \citep{EdelmanOS07,Varian07}.
In these auctions, advertisers bid for the opportunity to display their ad creatives alongside organic content.
Many of the web formats such as text, banners, video, apps, ... have their own subtleties which led to the development of new auction tools to handle them. Our goal in this paper is to investigate auction mechanisms to support the emerging format of AI-generated content. More specifically, we explore the use of auctions as a tool for influencing the output of large language models (LLMs) \cite[e.g.,][]{brown2020language}.

We consider a situation where a certain space in the web (which could be a part of a webpage, an UI element of an AI-chatbot, the dialog of a certain character in a video or a game, etc.) is
designated for commercial content and different advertisers can bid to influence the content in that space. Each advertiser has an LLM that can generate content for that space, and is willing to pay a certain amount of money for the right to have their content displayed. A simple design is to collect bids from advertisers and let the highest bidder choose whatever content they wish to publish in that space. While simple, this design does not exploit the flexibility of LLMs which is to combine different concepts in a creative way.

Consider this example. First, we ask an LLM to produce different ads for the fictitious Stingray Resort and the equally fictitious Maui Airlines:
\begin{itemize}
    \item ``\emph{Experience the magic of Hawaii at Stingray Resort, where stunning views, luxurious accommodations, and endless activities await. Book your stay today and create unforgettable memories in the heart of paradise}.''
    \item ``\emph{Fly to Hawaii with Maui Airlines and experience the beauty of the Aloha State. We offer affordable flights to all the major islands, so you can start your Hawaiian vacation sooner. Book your flight today and let the island spirit take over!}''
\end{itemize}
For that use case, however, the LLM is flexible enough to produce a joint ad for both:
\begin{itemize}
\item \emph{``Fly to paradise with Maui Airlines and experience the magic of Hawaii at Stingray Resort. Stunning views, luxurious accommodations, and endless activities await. Book your dream vacation today and create unforgettable memories.''}
\end{itemize}
One can envision an auction mechanism that allows both Stingray Resort and Maui Airlines to submit their LLMs and bids, with these inputs determining their prominence in the final outcome.\footnote{While this work's main focus is to create ad creatives that merge content from different advertisers, our designed auction mechanism  for merging LLM outputs could also be used in other contexts.}

\subsection{Unique Challenges}\label{sec:challenges}

LLMs \citep{brown2020language,lamnda,palm2} are an emerging technology with new and unconventional aspects, many of which have direct implications to auction design (e.g., how preferences are represented/expressed). Our goal is to identify some of the key challenges and take a first step in designing mechanisms to address them:

\begin{itemize}
\item {\bf Modelling and Expressing Preferences.} Auction theory typically models preferences via value functions that assign a value to each outcome. LLMs, however, as {\em generative models}, do not directly assign values. Instead, they succinctly encode preferences over outcomes within a stateless neural network model that predicts continuation probabilities.

\item {\bf Necessity of Randomization.} LLMs crucially rely on randomization. When forced to output tokens deterministically, LLMs often have a worse performance compared to situations that sample from a distribution~(see, e.g.,~ \citep{holtzman2019curious}, for a performance comparison of different decoding strategies). Therefore, an auction that aggregates LLM outputs should preferably also output distributions.

\item {\bf Technical Compatibility.}
Auction solutions should be compatible with current LLM technology, utilizing readily available information and integrating seamlessly.
Ideally, the allocation and payments should be obtained from simple manipulations of the LLM outputs.

\item {\bf Computational Efficiency.} LLM models are expensive to query, so the auction computation should not add too much overhead. In particular, auctions should not increase the number of calls to inference the models beyond the minimum necessary.
\end{itemize}

\subsection{Our Contributions}\label{sec:intro:model}

\paragraph{The Token Auction Model.}
Our first contribution is a formalism (``The Token Auction Model'')
for studying this problem. {\em Tokens} are the units making up sentences and paragraphs.\footnote{More generally, one can consider tokens forming parts of images~\citep{ramesh2021zero,yu2022scaling} and videos~\citep{sun2019videobert}. For the purpose of this paper, we will stick with text generation.} Examples of tokens include (sub-)words, symbols, numbers, and special tokens indicating the beginning and ending of the text. In particular, any piece of text (potentially incomplete) can be represented as an array of tokens, and any array of tokens also encodes a piece of text.

One salient feature of the state-of-the-art LLMs is that they are stateless, i.e., they maintain no internal memory or state. Instead, they simply map a prefix string to a distribution over the next token. The output is then created in an autoregressive manner. Given an input prompt, the output is generated by repeatedly feeding the current sequence of tokens into the LLM, sampling a continuation token, and appending it to the sequence of tokens.

The proposed \emph{token auction} operates on a token-by-token basis, and serves to aggregate several LLMs to generate a joint output. We assume the designer has access to algorithmic LLM agents represented by their respective text generation functions (the functions that map a sequence of tokens to a distribution over the next token). In addition, we allow each LLM agent to submit a single dimensional bid.
The auction output will be an aggregated distribution together with a payment rule that defines payments for each agent.\footnote{See our discussion later this section on the rationale of the indirect mechanism formulation.}

This approach may seem counterintuitive initially, as advertisers typically focus on the final generated text rather than individual word choices.
This seems to suggest a dynamic planning of the generated token sequence. However, existing LLMs do not reason about full pieces of text, nor do they plan ahead; instead, their preferences are expressed as desired distributions over merely the next token.
In other terms, we can think of an LLM as a succinct distillation of an agent's
complex combinatorial preferences over sequences of tokens into a generative token-by-token model.\footnote{See our discussion in Section~\ref{sec:aggregation_functions}, and Propositions~\ref{prop:kl_aggreg} and \ref{prop:rl_aggreg} for additional support for the stateless approach.}

The problem of aggregating LLMs forces the designer to understand the preferences of the agents away from the distilled LLM. This appears to be a very difficult problem.
Specifically, we believe it is implausible or at least impractical to assume an individual agent can meaningfully manipulate the distribution over tokens at any given stage, to direct the produced text to a more preferred one.
Our auction formulation seeks to strike a balance: By truthfully revealing the LLM to the designer, the agent gives the auction mechanism a hint as to what their preferred distribution is. The bids, in turn, can be used to tradeoff between agents, and in particular help the designer determine their relative weights.

\paragraph{Simple and Robust Token Auctions.}
Motivated by the challenges in modeling agents' preferences over generated distributions, we take a robust design approach aiming for token auctions that  provide desirable incentive properties, while imposing minimal assumptions on the agents' preferences over distributions.

Specifically, we model agents' preferences as entailing partial orders over distributions. Based on this partial preference order\footnote{Partial orders are more general than total orders, and hence our key results (such as Lemma~\ref{lemma:consistent}, Lemma~\ref{lemma:full-order}, and Theorem~\ref{thm:monotone}) apply to any complete preference order model.},
we formulate two desirable incentive properties, which we consider minimal requirements:
\begin{itemize}
\item \emph{Payment monotonicity}: Given two different bids by the same agent, a final distribution is closer to the desired distribution if and only if the payment is higher.
\item \emph{Consistent aggregation:} If for two different bids of the same agent, the final distribution is closer to the preferred distribution for some bids of the other agents, then it should be so for all bids of the other agents.
\end{itemize}

We show that any mechanism with these two properties is \emph{strategically equivalent} to a mechanism that satisfies a monotonicity requirement on the distribution aggregation function.

We then investigate whether it is possible to equip such distribution aggregation functions with payment rules that satisfy additional incentive properties. Specifically, we investigate whether such aggregation rules admit an analogue of the \emph{second-price payment rule}.
In the single-item second-price (or Vickrey) auction \citep{vickrey1961counterspeculation}, the payment
corresponds to the critical bid where an agent transitions from
losing to winning.
To port this notion to our setting, we show that under robust preferences (see Definition~\ref{def:obvious-prefer}) any monotone aggregation rule can be written
as a distribution over deterministic allocations from bids to tokens
such that there is a critical bid where the allocation transitions from a
less preferred to a more preferred token. Such a critical bid
then serves as a natural candidate for a payment rule. This  hence  leads to an analogue of the second-price auction  for our token auction model that only requires ordinal preferences. The resulting class of auctions is applicable whenever the agent valuations are compatible with the partial order, and provides robust incentives for all of these.

\paragraph{Designing Aggregation Functions.}

We then move to designing concrete aggregation functions. Our approach considers \alfs inspired by state-of-the-art LLM training, and   derives optimal distribution aggregation functions that minimizes such \alfs.

We focus on specific forms of \alfs based on KL-divergence, a commonly used loss function in LLMs. We consider two natural formulations inspired by current LLM training,  and show that the corresponding optimal aggregation rules are the weighted (log-space) convex combination of the target distributions from all participants.

The linear and log-linear aggregation rules we identify have different pros and cons. Both share the advantage that they are
optimal
for the respective \alfs.
The linear rule turns out to be monotone with respect to robust preferences, and is therefore compatible with the robust incentives approach. However, the log-linear rule is not.

\paragraph{Demonstration.} 
We conclude with demonstrations to support our token auction formulation, obtained by prompt-tuning of a publicly available LLM. A two-advertiser demonstrative example is considered, under both the linear and log-linear aggregation rules.
We show how the combined output varies as a function of $\lambda = \nicefrac{b_1}{(b_1+b_2)}$, where $b_1$ and $b_2$ are the advertisers' bids. Both approaches lead to meaningful and interpretable texts that smoothly transition from favoring one to favoring another advertiser, with a joint ad produced for intermediate values of $\lambda$.

\paragraph{Discussion/Design Choices.} An alternative to our approach of designing an \emph{indirect mechanism} would be to aim for a \emph{direct mechanism}. Such a mechanism, instead of asking agents for a scalar bid along with query access to the agents' LLMs,
would elicit the agents' full preferences directly.
However, this appears unrealistic in our new domain due to multiple reasons: (1) Allocation outcomes in our setting are a high-dimensional distribution,
whereas a classic mechanism's allocation is typically a subset of items, and often a single item in tractable setups. (2) While it is reasonable in the classic setup to elicit a valuation for an item or a subset of items, it does not appear realistic to elicit a high-dimensional utility function over all possible token distributions.
(3) Eliciting full preferences over any token distribution would require solving a problem that is strictly harder than what current LLMs are trained to do (namely, merely output the most preferred distribution). This level of complexity might go beyond current technological capabilities and would likely be computationally inefficient.

\subsection{Additional Related Work} 
To the best of our knowledge, the exact research question and our approaches in this work have not been previously studied. However, our work is indeed connected to a few lines of research.

\paragraph{Related LLM Research.}

Our work 
shares some similarities with the literature on fine-tuning LLMs,  with reinforcement learning from human feedback (RLHF) as a representative approach \citep{wei2021finetuned,bakker2022fine,ouyang2022training,bai2022training}. At a high level, fine-tuning and RLHF seek to align a generally pre-trained LLM with certain desirable behaviors. This is in spirit analogous to our goal of designing LLMs to better align with a group of agents' overall preferences. However, our research challenges and methods are both different from those in the fine-tuning literature. Specifically, fine-tuning refines the underlying model's parameters whereas our approach is one-layer up and directly aggregates the token distributions from multiple models.  The main challenge we address is the potential incentive misalignment while eliciting LLM agents' preferences, whereas human labelers or other models that generate reward feedback for RLHF are assumed to be genuine and do not misrepresent their own preferences.

The literature on in-context learning \citep{brown2020language,wei2022chain,wei2023larger} is similar to us in the sense that this approach also does not change the model parameters. A main difference to our work is that this literature seeks to influence token distributions by conditioning on better-generated prefix contexts, whereas we directly aggregate distributions from multiple LLM agents.

\paragraph{Connections in Mechanism Design.}
Our work is related to the literature on (combinatorial) public projects  \citep{papadimitriou2008hardness,dughmi2011truthful}. The connection is that one can view the output of the aggregated LLM in our situation as a public project that benefits the agents to different degrees.
Similar to these earlier studies, a core challenge in our problem is to elicit preferences about the public project from unknown agents. However, the design problem in our case is fundamentally different --- we choose a high-dimensional distribution from an $\R^T$ space with only partial knowledge about agents' preferences, whereas previous work has focused on the problem of choosing from a discrete (often exponentially large) set with clear agent valuation functions \citep{papadimitriou2008hardness,dughmi2011truthful}.

Another related stream of work includes \citep{freeman2019truthful,goel2019knapsack}, which studies the problem of  truthfully aggregating budget proposals. Their mechanisms
output a distribution over budgets that best serves the population, just like our mechanisms output distributions over tokens. However, the objectives and techniques between our work and theirs are both different. First, their problem is mechanism design without money, whereas our problem has monetary transfers involved. A direct consequence of this first difference is that their mechanisms will treat every participant with equal weight, whereas the weights of our participants are determined by their bids. Second, the research on truthful budget proposal aggregations typically assumes explicit valuation functions (e.g., $l_1$ distance between preferred and output distributions), under which the VCG mechanism is truthful. Their main research question hence is to study additional properties of the mechanisms such as Pareto-efficiency and certain fairness properties \citep{freeman2019truthful}. Assuming such an explicit valuation function does not appear realistic in our problem, so our core research question is to design robust mechanisms that enjoy good incentive properties simultaneously for a broad range of valuation functions.

From this perspective, our work also bears some similarity to the rich literature on robust mechanism design. Most of this literature still assume existence of value functions with uncertainty modeled by Bayesian beliefs or in a max-min sense \citep{bergemann2005robust,bergemann2012robust,roughgarden2016optimal,carroll2015robustness,DuttingRT19}. However, assuming such a valuation function over tokens or their distributions does not appear realistic in creatives generation, thus our model is more similar to a worst-case style consideration during which we only assume partial (``obvious'') preferences.

\paragraph{Follow-Up Work.}  Several papers follow-up on our work, by studying mechanism design problems for LLMs. \cite{DubeyEtAl24} consider bidders that bid for placement of their content within a summary generated by a large language model. \cite{SoumaliasEtAl24} design a truthful mechanism that generates several samples from a reference LLM, and incentivizes bidders to truthfully reveal their preferences. \cite{mordo2024sponsored} consider sponsored question answering, in which an organic answer to a search query is fused with an ad to create a sponsored answer, and advertisers bid on the sponsored answers.

\section{Preliminaries}

In this section, we first provide an abstraction of typical generative models and then introduce the basic formalism of the mechanism design problem we study.  For concreteness we adopt a terminology that suits the important LLM use case where the creative is text.

\subsection{Abstraction of Large Language Models}
{\em Large language models} (LLMs) \citep{brown2020language,lamnda,palm2} can be abstracted as functions mapping from a partial sentence to the distribution of the next \emph{token} that extends the partial sentence.

Formally, let $T$ be the set of tokens and $\Delta(T)$ be the set of distributions over $T$. Let $T^* = T \cup T^2 \cup \cdots \cup T^K$ denote the set of sequences of tokens, where $K$ is the maximum sequence length that the LLM can handle. Each LLM is modeled as a function $f: T^* \rightarrow \Delta(T)$ that maps any sequence of tokens to a distribution over the next token.

\paragraph{Autoregressive Text Generation.}
A \emph{prompt} is an initial set of tokens $s_0 \in T^*$ provided with instructions of what text to generate. An LLM produces a text in response to the prompt by sampling a token $\tau_1 \sim f(s_0)$ and constructing $s_1 = s_0 \oplus \tau_1$ (where $\oplus$ is the operation to append a token to an array). We then repeat the process of $\tau_k \sim f(s_k)$ and $s_k = s_{k-1} \oplus \tau_k$ until a special \emph{end-of-sentence} token is sampled. If at some point the sequence of tokens becomes too long (larger than $K$) we trim $s_k$ to its length-$K$ suffix.

Note that LLMs are stateless: They lack internal memory beyond the generated token sequence, and each token is sampled independently.

\paragraph{Training of LLMs.~} An LLM $f$ is parameterized by a neural network structure $M$ and a set of weights $W$. The weights are often obtained by three stages of optimization (see first three rows in Table~\ref{tab:train_stages}). The initial stage is very computationally intensive but task independent. Subsequent stages are less costly and their goal is to adapt the general purposed model obtained in the first stage to more specific tasks. Each of the stages usually minimizes a different loss function over a different dataset. The details of the training process are not particularly relevant to our discussion, though  a more detailed discussion can be found in Section~\ref{sec:llm_training}. We will note, however, that some of the mechanisms we discuss for combining the inputs of different LLMs resemble the   functional forms used in the reinforcement learning and fine-tuning steps.

\begin{table*}[t]
    \centering
    {
    \tableonesize
    \begin{tabular}{|c|c|c|c|}
      \hline
      Training/Learning stages & Data & Cost & Goal  \\
      \hline
      $\text{Pre-training}^\ast$ 
      & General texts from web, books, etc & Very high & A common baseline shared across downstream tasks   \\
      \hline
      $\text{Instruction fine-tuning}^{\dagger}$ 
      & Task specific data & Medium & Optimize the behavior for specific tasks  \\
      \hline
      $\text{RLHF}^{\ddag}$
      & Human evaluations & Medium & Security control, reducing harmful behavior, etc  \\
      \hline
      In-context few-shot learning & Carefully designed prompts as inputs & Very low & Effectively influence the behavior in real-time  \\
      \hline
    \multicolumn{4}{c}{}\\[-4pt]
    \multicolumn{4}{c}{\scriptsize{$\ast$~\cite{brown2020language,palm2}~ $\dagger$~\cite{wei2021finetuned}~ $\ddag$~\cite{ouyang2022training}}}
    \end{tabular}
    }

    \caption{Common Training Stages of LLMs.}

    \label{tab:train_stages}
\end{table*}
\rollbackfootnotecounter

\subsection{Token Auctions for LLMs}

We now formalize the mechanism design problem of combining the output of different LLM-represented algorithmic agents. As discussed in the introduction, we will design an auction to act on the token-by-token generation stage. Our goal is to keep the auction technically aligned with the state-of-the-art LLM systems.

\paragraph{Robust Modeling of LLM Agents' Preferences.~}
A key challenge in designing mechanisms for LLM agents is comparing their ``utilities'' over different output distributions.
To illustrate, suppose an LLM agent's preferred distribution over two tokens is $p=(0.6, 0.4)$, and consider two possible generated distribution outcomes:  $q_1 = (0.5, 0.5)$ and $q_2 = (0.8, 0.2)$. Between $q_1$ and $q_2$, it is unclear which one this LLM agent would prefer since while $q_2$ appears more distant from $p$ than $q_1$, it has a higher probability on the first token which appears more preferably by the LLM's initial distribution $p$.

Despite this incomparability between $q_1$ and $q_2$, it does appear clear that $q_2$ would be less preferred by the LLM than $q_3 = (0.7, 0.3)$. This is because $q_3$ deviates from $p$ along the same directions as $q_2$ for each entry (i.e., both increase or both decrease), but deviates less in terms of the absolute value of deviation.

The above observation illustrates that while it is difficult to model LLM agents' complete preferences over all the generated distributions, it seems natural to assume certain \emph{partial order} over the distributions. This motivates us to consider a \emph{robust} modeling of LLM agents' preferences through general partial orders, and sometimes, the following specific notion of \emph{robust preferences}.
\begin{definition}[Robust Preferences over Distributions]\label{def:obvious-prefer}
Consider any LLM agent $i$ with preferred distribution $p_i$, and any two aggregation distribution $q, q' \in \Delta(T)$. We say $q$ is (weakly) \emph{robustly preferred} over $q'$ by agent $i$, or formally, $q \succeq_i q'$, if
\begin{eqnarray}
    \forall t \in T, &   \, \,  \, \, \, \,    |q(t) - p_i(t)| \leq |q'(t) - p_i(t)| \\
   & \text{and }  \, \, \, \,        (q(t) - p_i(t))(q'(t) - p_i(t)) \geq 0.
\end{eqnarray}
Moreover, if $q \neq q'$, then $q$ is strictly preferred over $q'$ by $i$,  i.e., $q \succ_i q'$.
\end{definition}
In other words, $q$ is robustly preferred by $i$ over $q'$ when (1) the deviation  of $q$ from $p_i$ is smaller than the deviation of $q'$ from $p_i$ for every entry; and (2) these deviations are along the same direction for every entry. Note that Definition \ref{def:obvious-prefer} only specifies a  partial ordering among aggregated distributions.  Thus it is possible that two distributions are not comparable, i.e.,   $q \not\succeq_i q'$ and $q' \not\succeq_i q$.

\paragraph{Token Auctions.} Our goal is to design simple, practical auction mechanisms that work well under minimal assumptions about the agents' private preferences. Specifically, we seek to design \emph{token auction mechanisms} $\mathcal{M} = \langle q, z \rangle$, where $q$ is a distribution aggregation function and $z$ is a payment function. A token auction mechanism operates on a token-by-token basis, and lets $n$ algorithmic LLM agents influence the output distribution and payments through scalar bids. The bid profile of all agents is denoted as $\b = (b_1, \ldots, b_n) \in \Bids^n_+$ where any bid $b_i \in \Bids_+$ is a non-negative rational number.
We assume that the initial prompt $s_0 \in T^*$, and the text aggregation functions $f_1, \ldots, f_n$ of the $n$ LLM agents are publicly known.

\emph{Distribution Aggregation Function.}
This is the first ingredient to a token auction mechanism.
A distribution aggregation function $q$ takes as input a vector of bids $\b \in \Bids^n_+$ and $n$ distributions $\p \in \Delta(T)^n$ and maps these to a distribution over tokens:
\[\text{aggregation function:} \quad  q: \Bids^n_+ \times \Delta(T)^n \rightarrow \Delta(T).\]

For fixed bids, a distribution aggregation function can be used in the same way as a text aggregation function. Namely, starting from the initial prompt $s_0 \in T^*$, we can repeatedly sample $\tau_k$ from distribution $q_k = q((b_1,\ldots,b_n),(f_1(s_{k-1}),\ldots, f_n(s_{k-1})))$  for each $k \geq 1$ to generate $s_k = s_{k-1} \oplus \tau_k$.
Note the alignment with LLMs, which already produce the  distributions $f_i(s_{k-1})$ for $i \in [n]$. No additional calls to the LLMs are needed.

\emph{Payment Function.} In addition to the distribution aggregation function, we seek to design payment functions. Here, we want to operate on a token-by-token basis and seek a stage-independent design akin to the stage independence of state-of-the-art LLMs' token generation.
Formally, for each agent $i$, we aim to define a
\begin{gather*}\text{pricing function:}\quad \zeta_i: \Bids^n_+ \times \Delta(T)^n \times T \rightarrow \R,
\end{gather*}
with the interpretation that for bids $\b \in \Bids^n_+$, distributions $\p \in \Delta(T)^n$, and token $t \sim q(\b,\p)$, the payment from agent $i$ is
$\zeta_i(\b, \p, t)$.
These pricing functions naturally lead to expected payments by taking expectations over tokens. Namely, for each agent $i$, we define
\begin{gather*}
\text{payment function:}\quad z_i: \Bids^n_+ \times \Delta(T)^n \rightarrow \R,
\end{gather*}
as the function that takes as input a vector of bids $\b \in \Bids^n_+$ and distributions $\p \in \Delta(T)^n$, and maps these to $z_i(\b,\p) = \E_{t \sim q(\b,\p)}[ \zeta_i(\b,\p,t)]$.

\paragraph{Discussion.}

Token auctions provide a suitable abstraction for analyzing the strategic aspects of LLM aggregation. Fully expressing agent preferences over all generated content is impractical. Representing agents as LLMs is a plausible approach, as LLMs distill preferences into token distributions. Thus, auctioning tokens based on LLM-expressed preferences is a natural mechanism.

At the same time, the detailed functioning of LLMs remains rather opaque, and it seems implausible that agents could meaningfully misreport the outcome distributions of their LLMs in order to achieve a more desirable aggregated output.  Our auction formulation offers a middle ground. We assume the designer has access to the LLMs, 
but let the agents influence the aggregation process through a single dimensional bid.

\section{Incentives in Token Auctions}\label{sec:detail-free}

In this section, we examine the strategic properties of token auctions. Our goal is robust incentive properties that rely on as few assumptions about the agents' preferences as possible.
We first formulate two natural properties that any reasonable mechanism should satisfy, and show that they are essentially equivalent to a monotonicity requirement on the distribution aggregation function. We then show that, when agents have robust preferences (see Definition~\ref{def:obvious-prefer}), for any such monotone distribution aggregation function, it is possible to define a natural second-price payment rule.

\subsection{Desirable Incentive Properties}
We begin by formulating two conditions that reasonable token auction mechanisms should satisfy with respect to general partial orders $\succeq_i$.
The first is a \emph{monotonicity} condition on the payment function. It requires that agents' pay increases if and only if they obtain better distributions.

\begin{definition}[Payment Monotonicity]\label{def:mono-pay}
Mechanism $\M = \langle q, z \rangle$ satisfies \emph{payment monotonicity},
if for all $\p$, $\b_{-i}, b_i, b'_i$
we have
\[
z_i(b_i, \b_{-i}, \p) \geq z_i(b'_i, \b_{-i}, \p)
\iff
 q(b_i, \b_{-i}, \p)  \succeq_i  q(b'_i, \b_{-i}, \p).
\]
\end{definition}

It is a natural incentive constraint because if the payment function is not monotone, then bidders will likely manipulate their bids in order to induce better distribution with lower payment. The following is a useful implication of payment monotonicity. Given $\p, \b_{-i}$, if $i$'s two bids $b_i, b'_i$ lead to the same aggregated distribution, i.e., $q(b_i, \b_{-i}, \p)  =  q(b'_i, \b_{-i}, \p)$, then as a consequence of payment monotonicity the payment must also be the same.

The second incentive constraint 
is about the consistency of the aggregation function.
Intuitively, the relative preference between two bids should not be influenced by the bids of other agents.

\begin{definition}[Consistent Aggregation]\label{def:consistent-aggregate}
The distribution aggregation function $q(\b, \p)$ is said to be \emph{consistent} if it admits consistent ordering across all $\b_{-i}$. Formally, 
if $q(b_i, \b_{-i}, \p)   \succ_i  q(b'_i, \b_{-i}, \p)$ for some $\b_{-i}$, then for all $\b'_{-i}$, $q(b_i, \b'_{-i}, \p) \succeq_i q(b'_i, \b'_{-i}, \p)$.
\end{definition}
Similar to payment monotonicity, this requirement of consistent aggregation is imposed to avoid bidders' concerns that the same bid can lead to better or worse distributions, depending on the opponents' bids.
To rule out such behavior, consistency requires that whenever bids $b_i,b'_i$ are such that $b_i$ is strictly preferred over $b'_i$ for some opposing bids $\b_{-i}$, then it better be that $b_i$ is weakly preferred over $b'_{i}$ for all opposing bids $\b'_{-i}$.

\subsection{Monotone Aggregation Functions}\label{sec:incentive:mono}

Next we show a ``revelation principle'' type of result, stating that if one is interested in mechanisms satisfying the desirable incentive properties stated above (Definition~\ref{def:mono-pay} and Definition~\ref{def:consistent-aggregate}), then one can without loss of generality focus on \emph{monotone aggregation functions} as captured in the following definition for any partial order $\succeq_i$.

\begin{definition}[Monotone Aggregation Function] \label{def:mon-agg}
The distribution aggregation function $q(\b, \p)$ is called \emph{monotone} if any higher bid from any agent $i$ leads to a more preferred aggregated distribution for $i$. Formally, for all $\p$, $\b_{-i}$ and $b_i \geq b'_i$:
$$q(b_i, \b_{-i}, \p) \succeq_i q(b'_i, \b_{-i}, \p).$$
\end{definition}

We are now ready to state our main finding in this subsection with the following definition of \emph{strategic equivalence} from one mechanism $\M$ to another $\tilde \M$.  In words, the aggregated distribution outcome and all agents' payments will be the same under mechanism $\M$ and $\tilde \M$ after each agent $i$ applies some strategy mapping $\pi_i$. Formally,
\begin{definition}[Strategic Equivalence]
A mechanism $\M = \langle q, z \rangle$ is \emph{strategically equivalent} to another mechanism $\tilde \M = \langle \tilde q , \tilde z \rangle$, if $\forall \p \in \Delta(T)^n$, there exists a profile $\pi$ of strategy mappings with $\pi_i: \Bids_+  \to \Bids_+$ for every agent $i$ (i.e., $\pi(\b) = (\pi_1(b_1), \ldots, \pi_n(b_n))$), such that $\forall \b \in \Bids_+^n,~q(\b , \p) = \tilde  q(\pi(\b), \p)$ and $z(\b , \p) =  \tilde z (\pi(\b), \p).$
\end{definition}

\begin{theorem}[Revelation Principle]\label{thm:monotone}
Any mechanism $\M = \langle q, z \rangle$  with a consistent distribution aggregation function $q$ and a monotone payment function $z$ is strategically equivalent to a mechanism $\tilde \M = \langle \tilde q , \tilde z \rangle$ which has a monotone distribution aggregation function $\tilde q$ and a monotone payment function $\tilde z$.
\end{theorem}

Theorem~\ref{thm:monotone} can be viewed as a revelation principle in the sense that it simplifies the design choice of aggregation functions. Monotone aggregation functions are a strict subset of consistent aggregation functions since monotonicity directly implies a total order over possible aggregated distributions $Q(\b_{-i}, \p) = \{q(b_i, \b_{-i}, \p): b_i \in \Bids_+\}$, with the order naturally determined by the real-numbers' order on $i$'s bid $b_i$  hence this order will be consistent across different $\b_{-i}$ and $\p$. In this sense, one might think that consistency --- which does not impose any restriction when $q(b_i, \b_{-i}, \p)$ is not strictly preferred over $q(b'_i, \b_{-i}, \p)$ --- might be a significantly weaker requirement on  aggregation functions than monotonicity which requires a total and consistent order.
Theorem~\ref{thm:monotone} shows that this is not the case --- they are essentially the same as long as the natural incentive requirement of payment monotonicity is also imposed.

\medskip

The proof of Theorem~\ref{thm:monotone} hinges on the following two lemmas,  Lemma~\ref{lemma:full-order} and Lemma~\ref{lemma:consistent}. Together these two lemmas imply the existence of a strategy mapping, under which the resulting aggregation function becomes monotone. The proof of the theorem is completed by applying the same mapping to the payment function, and noting that this ensures payment monotonicity.
\hybridtext{We refer the readers to the \href{https://arxiv.org/abs/2310.10826}{full version} for the formal proofs}{We defer the formal proofs of these results to Appendix \ref{appendix:proof-mono}}.

\begin{lemma}\label{lemma:full-order}
  For any distribution aggregation function $q$, there exists a payment function $z$ such that mechanism $\M = \langle q, z \rangle$ is payment-monotone if and only if $\succeq_i$ establishes  a total order over $Q(\b_{-i}, \p) = \{q(b_i, \b_{-i}, \p): b_i \in \Bids_+\}$ for any fixed $\b_{-i}$ and $\p$.
\end{lemma}

\begin{lemma}\label{lemma:consistent} Consider any consistent aggregation function $q$. Suppose $\succeq_i$ defines a total order over the aggregation set $Q(\b_{-i}, \p) = \{q(b_i, \b_{-i}, \p): b_i \in \Bids_+\}$ induced by agent $i$'s bid for any fixed $\b_{-i}$ and $\p$, then there exist a profile $\pi$ of strategy mappings such that $q(\b, \p) = \tilde q(\pi(\b), \p)$ for some monotone aggregation function $\tilde q(\cdot, \cdot)$.
\end{lemma}

We conclude this subsection by giving two natural examples used in today's machine learning practice: an example of a monotone aggregation function and a non-monotone one, both with respect to robust preferences (see Definition~\ref{def:obvious-prefer}).

\begin{example}[Linear Aggregation] Consider  $q_\mathsf{KL}(\b, \p)$ defined as $$\textstyle q_{\mathsf{KL}} = \frac1B\sum_{i \in [n]} b_i \cdot p_i,~\text{where}~B =  \sum_{i \in [n]} b_i.$$
It is easy to verify that this is a monotone aggregation function with respect to robust preferences.

\end{example}

\begin{example}[Log-linear Aggregation]\label{exmpl:log-linear}
Consider the aggregation function $\bar q_{\mathsf{KL}}(\b, \p)$ defined by the following equations:   \[\textstyle \forall t \in T,~\ln \bar q_{\mathsf{KL}}(t) = \frac1B \sum_{i \in [n]} b_i \cdot \ln p_i(t) - C,\]
where $B =  \sum_{i \in [n]} b_i$ and $C =  \ln \sum_{t \in T} e^{\frac1B \sum_{i \in [n]} b_i \cdot \ln p_i(t)}$.

The following two-agent example shows that $\bar q_{\mathsf{KL}}$ is not monotone: $p_1 = (.5, .4, .1)$ and $p_2 = (.5, .1, .4)$. When $b_1 = b_2$, $\bar q_{\mathsf{KL}} = (\sqrt{.25}, \sqrt{.04}, \sqrt{.04}) / .9 = (5 / 9, 2 / 9, 2 / 9)$. Fix $b_2 = 1$, either $b_1 \rightarrow 0$ or $b_1 \rightarrow \infty$, $\bar q_{\mathsf{KL}}(t_1) = .5 < 5 / 9$. Hence $\bar q_{\mathsf{KL}}$ is not monotone with respect to robust preferences.
\end{example}

\subsection{Second Price Payment Rules}\label{sec:spa}

Next, we investigate whether monotone aggregation functions admit a ``second-price'' payment rule, capturing the Vickrey auction's concept of the ``minimum bid to win''. In the Vickrey auction, the payment corresponds to the critical bid where an agent transitions from losing to winning. To port this notion to our setting, we will show in Theorem~\ref{thm:stable} that, under robust preferences, any monotone aggregation rule can be written as a distribution over deterministic allocations from bids to tokens such that there is a critical bid where the allocation transitions from a less preferred to a more preferred token. Such a critical bid becomes then a natural candidate for the payment. We will show that besides being payment monotone it has other desirable properties, such as a Myerson-style characterization \citep{myerson1981optimal} in terms of the total variation distance between the preferred distribution and the outcome.
To define our payment rule, we first introduce the notion of stable sampling.

\subsubsection{Stable Sampling}\label{subsec:stable}

To design the payment for bidder $i$, we analyze a monotone distribution aggregation function $q(\b, \p)$ from $i$'s perspective, fixing the distributions $\p$ and the bids of other agents $\b_{-i}$. To simplify the notation, we refer to $q(b_i, \b_{-i}, \p)$ as $q(b_i)$.

We define an implementation of $q(\cdot)$ as a function $\sigma$ that maps $b_i$ and an exogenous random variable $r \sim \mathcal{R}$ (independent of $\b$ and $\p$) to a token $t \in T$. Here, $r$ is the \emph{random source} used for implementing the random sampling of token $t$. Hence when $r$ is fixed, $\sigma(b_i, r)$ is fully deterministic. We say that $\sigma: \Bids_+ \times \mathcal{R} \rightarrow T$ is a valid implementation of $q(b_i)$ if:
\[\textstyle \Pr_{r \sim R}[\sigma(b_i, r) = t] = q_t(b_i), \forall t \in T.\]

Next we define what it means for an implementation to be a ``stable'' sampling. It is useful to partition the token set $T$ into $\Tunder = \{t \in T: q_t(0) \leq (p_i)_t\}$ and $\Tover = \{t \in T: q_t(0) > (p_i)_t\}$ corresponding to undersampled ($\Tunder$) and oversampled ($\Tover$) tokens.
The monotonicity of aggregation function $q$ turns out to be equivalent to the monotonicity of $q_t(b_i)$, as formalized in the following lemma (proof in Appendix~\ref{app:proofs-spa}).

\begin{lemma}\label{lemma:monotone-alter}
Under agents' robust preferences, a  distribution aggregation function $q$ is monotone, if and only if for every agent $i$ and $b_i \in \R_+$,
\begin{enumerate}
\item $\forall t \in \Tunder$, $q_t(b_i) \leq (p_i)_t$ and $q_t$ weakly increases in $b_i$;
\item $\forall t \in \Tover$, $q_t(b_i) \geq (p_i)_t$ and $q_t$ weakly decreases in $b_i$.
\end{enumerate}
\end{lemma}

We can now define the key notion of stable sampling, and state and prove our main result below.

\begin{definition}[Stable Sampling]\label{def:stable-sampler}
  Let $q_i(b_i)$ be an aggregation function obtained by fixing $\b_{-i}$ and $\p$, and let $\Tunder$ and $\Tover$ be the sets of undersampled and oversampled tokens for agent $i$. Then we say that an implementation $\sigma$ is stable with respect to aggregation function $q$ if for any $r \in \mathcal{R}$ there are two tokens $u_r \in \Tunder$ and $o_r \in \Tover$ and a threshold $\theta_r \in \Bids_+ \cup \{\infty\}$ such that:
  \[
    \sigma(b_i,r) = o_r,~\text{if}~b_i < \theta_r;~\text{and}~\sigma(b_i,r) = u_r,~\text{if}~b_i \geq \theta_r.
  \]
\end{definition}

Here is an illustrative example of stable sampling for a simple distribution supported on two tokens:  one undersampled token $u$ with probability $q_u(b_i)$ and an overampled token $o$ with probability $q_o(b_i) = 1 - q_u(b_i)$. Monotonicity of $q$ implies  $q_u(b_i)$  weakly increases in $b_i$ and $q_o(b_i)$  weakly decreases. Sample $r$ uniformly from $[0,1]$. For any $b_i$, if $  q_o(b_i) > r $ assign $o$; assign $u$ otherwise. Note that this implementation is: (i) deterministic for fixed $r$; (ii) transits from token $o$ to $u$ as $b_i$ increases (hence $q_o(b_i)$ decreases) under any fixed $r$; (iii) matches the probabilities of $q(b_i)$ in expectation for any $b_i$. The following theorem generalizes this to any number of tokens (proof in Appendix~\ref{app:proofs-spa}).

\begin{theorem}\label{thm:stable} Given a monotone distribution aggregation function $q$ then for any agent $i$ with robust preferences and fixed $\b_{-i}$ and $\p$ there always exists a stable implementation $\sigma$ of $q$.
\end{theorem}

\subsubsection{A Second Price Rule via Stable Sampling}

A stable implementation of a monotone aggregation rule naturally leads to a second-price payment rule: for any given randomness $r$, if the oversampled token $o_r$ is sampled, the agent pays zero as it is the same token that would have been sampled under randomness $r$ if their bid was zero. If the agent's bid was high enough to move from the oversampled $o_r$ to the undersampled token $u_r$, then the agent pays the critical bid $\theta_r$.

Perhaps surprisingly, the expected payment induced by the second price rule described above does not depend on the actual implementation of the sampling algorithm so long as it is stable in the sense of Definition \ref{def:stable-sampler}. The expected payment is proportional to the extent to which the sampling shifts the distribution $q(b_i)$ towards the desired distribution $p_i$ as $b_i$ increases and,   interestingly, admits an expression simillar to the standard Myersonian payment as a function of allocations \citep{myerson1981optimal} in which the allocation probability is replaced by the total variation distance between the agent's preferred distribution $p_i$ and the implemented distribution $q(b_i)$.

\begin{proposition}
Consider an arbitrary stable sampling implementation of $q(b_i)$ and the induced second price rule using the critical bid $\theta_r$ as described above. Then the expected payment satisfies     $$ z_i(b_i) =
\frac{1}{2} \int_0^{b_i} \big[ \Vert q(b_i) - p_i\Vert_1 - \Vert q(b') - p_i\Vert_1 \big] \mathrm{d}b', \qquad \forall b_i \in \Bids_+.  $$
\end{proposition}
\begin{proof}
We again omit the terms $\b_{-i}, \p$ since they are fixed in each context. If a token $t \in \Tover$ is sampled, the payment is naturally $\zeta_i(b_i, t)=0$. For a token $t \in \Tunder$ we have:
\begin{align*} & \zeta_i(b_i, t) q_t(b_i) = \textstyle \E_r[\theta_r \cdot \mathbf{1}\{\sigma(b_i,r)=t\}] \\ = & \textstyle \E_r \int_0^{b_i} \big[ \mathbf{1}\{\sigma(b_i,r)=t\} - \mathbf{1}\{\sigma(b',r)=t\} \big]  \mathrm{d}b' \\
=& \textstyle \int_0^{b_i} [q_t(b_i) - q_t(b')] \mathrm{d}b'.
\end{align*}
Hence:
\begin{align*}  z_i(b_i) & \textstyle = \sum_t   \zeta_i(b_i, t) q_t(b_i)  \textstyle = \int_0^{b_i} \bigg[ \sum_{t \in \Tunder} q_t(b_i) - \sum_{t \in \Tunder} q_t(b') \bigg] \mathrm{d}b'\\ & \textstyle =
\frac{1}{2} \int_0^{b_i} \big[ \Vert q(b_i) - p_i\Vert_1 - \Vert q(b') - p_i\Vert_1 \big] \mathrm{d}b',
\end{align*}
as claimed.
\end{proof}

\paragraph{Counterfactuals.} The practical advantage of a stable sampling implementation coupled with a second price rule is to offer advertisers a description where it is clear that they only pay if they can improve the outcome.
Moreover, with a fixed $r$, advertisers can easily evaluate counterfactuals, as each token can only have one of two possible outcomes, simplifying the comparison.

\paragraph{Universally Stable Sampling.}
We chose to define stable sampling as a single-agent algorithm with fixed $\b_{-i}$.
One may define a universal stable implementation as $\sigma^{\mathsf{univ}}:\Bids^n \times \mathcal{R} \rightarrow T$ such that
\begin{gather*}
  \textstyle
  \forall \b \in \Bids^n_+, t\in T,~\Pr_{r \sim R}[\sigma^{\mathsf{univ}}(\b, r) = t] = q_t,
\end{gather*}
and for any $i$ and $\b_{-i}$, $\sigma^{\mathsf{univ}}(\cdot, \b_{-i}, r)$ is stable. In Appendix~\ref{app:universally-stable},
we provide counter-examples where such $\sigma^{\mathsf{univ}}$ do not always exist.

\input{www/aggregation-functions}

\input{www/demo}

\hybridtext{}{\input{www/conclusion}}

%% file: www/aggregation-functions.tex
\section{Design of Aggregation Functions}\label{sec:aggregation_functions}

In the previous section, we discussed payment schemes and incentive properties assuming we have an aggregation function. Here we investigate the design of aggregation functions. We adopt two guiding principles: (1) We first define \analf
to model the overall satisfaction of the agents with the final distribution $q$, weighted by their bids $b_i$. The \alf
has the form:
$$\textstyle \Loss(\p, \b, q) = \sum_i b_i \rho(p_i, q),$$
where $\rho : \Delta(T) \times \Delta(T) \rightarrow \R$ indicates how different the distribution $q$ is from the preferred $p_i$. (2) The second is to define the
function $\rho$ based on the typical loss functions in LLM training.

\subsection{Review of LLM Training}\label{sec:llm_training}

So far we have assumed that an LLM $f:T^* \rightarrow \Delta(T)$ is already trained. In order to discuss the training process, it is useful to recall that an LLM is a neural network parameterized by a vector of weights $W \in \mathbb{R}^N$ in a very high dimensional space. To discuss training, it is useful to think of $f$ as a function of both input and weights:
$$f:T^* \times \R^N \rightarrow \Delta(T).$$
We represent the second argument by a superscript $W$. Training is to optimize $W$ such that $f^W(\cdot)$ minimizes a certain loss function over a dataset. A dataset is a sequence of pairs $(x_i, y_i)$ with $x_i \in T^*$ (input sequence) and $y_i \in T$ (label), and a loss function is a function $\ell : T \times \Delta(T) \rightarrow \R$. A network typically seeks to minimize:
$$\textstyle \min_W \sum_i \ell(y_i, f^W(x_i)).$$
A widely used loss function in the first stage is the KL-divergence:
$$\ell_\mathsf{KL}(y, x) = - \ln [f^W(y \vert x)],$$
where we use the notation $f(y\vert x)$ to represent the probability that a token $y \in T$ is sampled from $f(x) \in \Delta(T).$

It is useful to think of this problem in the limit where the size of the dataset grows to infinity and it can be effectively thought of as a full-support distribution over $T^* \times T$. In that setting, we can represent the dataset as a distribution $\mu \in \Delta(T^* \times T)$ over input-label pairs $(x,y) \in T^* \times T$. We will write $\mu(x) = \sum_y \mu(x,y)$ to denote the marginal distribution on $x$ and $\mu(\cdot \vert x)$ to denote the conditional distribution of labels given an input $x$. Then
\begin{equation}\label{eq:KL_problem}\tag{KL}
\!\!\!\!\!\!\!\!\!\!\!\! \displaystyle{\min_W}  \textstyle \L_{\KL}^\mu(f^W), ~~~ \L_{\KL}^\mu(f) := {\sum_{x  \in T^*}} \mu(x) \cdot D_\mathsf{KL}(\mu(\cdot \vert x) \Vert f(x)),
\end{equation}
where
$D_{\mathsf{KL}}(p \| q) = \sum_t p(t) \ln \frac{p(t)}{q(t)}$
is the KL-divergence.

LLMs are typically trained through successive refinement of weights: $W^{\mathsf{pre}} \rightarrow W^{\mathsf{SFT}} \rightarrow W^{\mathsf{RL}}$. In pre-training we compute $W^{\mathsf{pre}}$ by solving problem \eqref{eq:KL_problem} on a generic dataset $\mu^{\mathsf{pre}}$ via stochastic gradient descent. In the second stage, we initialize the weights as $W = W^{\mathsf{pre}}$ and run stochastic gradient descent to solve the same problem \eqref{eq:KL_problem} on a more specialized dataset $\mu^{\mathsf{SFT}}$. In other words, we solve the same problem on two different datasets.

The problem in the RLHF stage is different. The dataset only contains $x$ (still represented by $\mu$) and for any $y$, we have a function $r(x,y)$ giving the reward of mapping $x$ to $y$. Then $W^{\mathsf{RL}}$ is obtained by maximizing reward while minimizing the distance to the function $f^{\mathsf{SFT}}$ from previous stages (the PPO algorithm \citep{schulman2017proximal,ouyang2022training}):
\begin{align}
& \max_W \L_\RL^{\mu,r}(f^W) \label{eq:RL_problem}\tag{RL}, \\
& \L_\RL^{\mu,r}(f) := \sum_{x  \in T^*} \mu(x) \Bigl[\sum_y r(x,y) f(y \vert x) - \beta D_\mathsf{KL}(f(x)\Vert f^{\mathsf{SFT}}(x)) \Bigr]. \nonumber
\end{align}

\paragraph{KL-divergence and Entropy.} For the next propositions it is useful to recall that the entropy of a distribution $p \in \Delta(T)$ is $H(p) = -\sum_{t\in T} p(t) \ln p(t)$. Given two distributions $p,q \in \Delta(T)$, the cross entropy of $q$ relative to $p$ is $H(p,q) = -\sum_{t\in T} p(t) \ln q(t)$. Hence we can write $D_{\mathsf{KL}}(p\|q) = H(p,q) - H(p)$. We will also use Gibbs' inequality $H(p) \leq H(p,q), \forall p,q \in \Delta(T)$.

\subsection{KL-inspired Aggregation}

The first aggregation method is based on the \eqref{eq:KL_problem} program. When trying to aggregate LLMs $f_i$ according to bids $b_i$, we will design a function $q$ that mimics the outcome of the following thought experiment. We will imagine that each LLM was obtained by solving the \eqref{eq:KL_problem} on a dataset represented by $\mu_i$, where the marginal over inputs are the same $\mu_i(x) = \mu(x), \forall i$ but potentially differ on the marginals on the labels $\mu_i(y \vert x)$. In this thought experiment, we will combine their LLMs by re-training an LLM on the combined labels weighted by the bids. In other words, we will solve the problem:
\begin{equation}\label{eq:kl_aggreg_1}
\textstyle \min_W \text{ }\L_{\KL}^{\bar \mu}(f^W), \quad \quad \text{where}~\bar \mu = \frac{\sum_i b_i \mu_i}{\sum_i b_i}.
\end{equation}

The next proposition morally says that we can obtain a  solution to the \eqref{eq:KL_problem} problem on the aggregated dataset \eqref{eq:kl_aggreg_1} by combining its solutions on individual datasets. The proof appears in Appendix~\ref{app:proofs-agg}.

\begin{proposition}\label{prop:kl_aggreg} Consider datasets $\mu_i$ such that $\mu_i(x) = \mu(x), \forall i,x$ and let $\bar \mu$ be their weighted average. Let $f_i$ be the minimizer of $\L_{\KL}^{\mu_i}$ and $f^*$ be the minimizer of $\L_{\KL}^{\bar \mu}$, the loss on the aggregated dataset. Then $f^*$ is the solution to:
\begin{equation}\label{eq:kl_aggreg_2}
\textstyle  \min_f \sum_x \mu(x) \sum_i b_i D_\mathsf{KL}(f_i(x) \Vert f(x)).
\end{equation}
\end{proposition}

Proposition \ref{prop:kl_aggreg} motivates the following \alf:
\[\textstyle \Loss_{\mathsf{KL}} = \sum_{i \in [n]} b_i \cdot \rho_{\mathsf{KL}}(p_i, q) = \sum_{i \in [n]} \sum_{t \in T} b_i \cdot p_i(t) \cdot \ln \frac{p_i(t)}{q(t)}.\]

Now, we characterize the aggregation function that minimizes $\Loss_{\mathsf{KL}}$. The proof in Appendix~\ref{app:proofs-agg} uses Gibb's inequality.

\begin{lemma}\label{lemma:efficient}
  The aggregation function that minimizes $\Loss_{\mathsf{KL}}$ is the linear combination of $p_i$:
  \[\textstyle \forall t \in T,~q_{\mathsf{KL}}(t) = \frac{\sum_{i} b_i \cdot p_i(t)}{\sum_i b_i}.\]
\end{lemma}

Besides being monotone and aligned with how LLMs are trained, this aggregation function has the advantage that we only need to call a single LLM to sample each token. We can choose an index $i$ proportionally to the bids and then sample a token from the $i$-th LLM. To compute second price payments, however, we still need to query the LLMs for all agents.

\subsection{RL-inspired Aggregation}

Now, consider a different thought experiment where all agents use the same pre-trained and fine-tuned model of weights $W^{\mathsf{SFT}}$, but employs a different reward function $r_i(x,y)$ for RLHF. We combine their LLMs by solving the \eqref{eq:RL_problem} problem on the combined reward functions weighted by the bids. We will solve the problem:
$$\textstyle \max_W \L_{\RL}^{\mu, \bar r}(f^W), \quad \bar r = \frac{\sum_i b_i r_i}{\sum_i b_i}.$$

Analogously to Proposition~\ref{prop:kl_aggreg}, we show that we can obtain the solution to the aggregated problem by combining the solutions on each dataset. We defer the proof of Proposition~\ref{prop:rl_aggreg} to Appendix~\ref{app:proofs-agg}.

\begin{proposition}\label{prop:rl_aggreg} Consider datasets $\mu, r_i$ and let $f^{\mathsf{SFT}}$ be the solution of program \eqref{eq:KL_problem} with data $\mu$. If $f_i$ is the maximizer of $\L_{\RL}^{\mu, r_i}$, let $f^*$ be the maximizer of $\L_{\RL}^{\mu, \bar r}$ where $\bar r$ is the weighted average of the reward functions, then $f^*$ is the function minimizing:
\begin{equation}\label{eq:prop_rl}
  \textstyle \min_f  \sum_x \mu(x) \sum_i b_i D_\mathsf{KL}(f(x) \Vert f_i(x)).
\end{equation}
\end{proposition}

Similar to what we did in the previous subsection, we can also define \analf inspired by Proposition \ref{prop:rl_aggreg}. Namely:

$$\textstyle \overline{\Loss}_{\mathsf{KL}} = \sum_{i \in [n]} b_i \cdot \rho_{\mathsf{KL}}(q,p_i) = \sum_{i \in [n]} \sum_{t \in T} b_i \cdot q(t) \cdot \ln \frac{q(t)}{p_i(t)}.$$

\begin{lemma}\label{lemma:efficient-log}
    The aggregation function that minimizes $\overline{\Loss}_{\mathsf{KL}}$ is the log-linear combination of $p_i$:
    \[\textstyle \forall t \in T,~ \ln \bar q_{\mathsf{KL}}(t) = C + \frac{\sum_i b_i \ln \cdot p_i(t)}{\sum_i b_i},\]
    where $C$ is a normalization constant such that $\sum_t \bar q_{\mathsf{KL}}(t) =1$.
\end{lemma}
The proof follows directly from the proof of Proposition~\ref{prop:rl_aggreg}.
While not monotone (Example~\ref{exmpl:log-linear}),
the log-linear aggregation function is a reasonable choice assuming that the agents' preferences are aligned with the KL-divergence loss for the RL-stage training.

%% file: www/demo.tex
\hybridtext{\input{www/demo-table}}{}

\section{Demonstration}\label{sec:demo}

We implement the aggregation functions proposed in Section~\ref{sec:aggregation_functions} and discuss the examples they produce. Off-the-shelf LLMs generate full text passages. In our case, we need to peek at the internal states of LLMs (the probability distributions over tokens) at each token generation stage.
Therefore, we use a custom version of the \modelname{} model with a modified inference method that allows access to the token distributions.

Starting from the same base model, we customize it for different agents by prompt-tuning. We start with a base model $f:T^* \rightarrow \Delta(T)$ and for each agent $i$ we come up with a ``prompt'' $s_0^i \in T^*$ and now we define for each agent $i$ the LLM $f_i:T^* \rightarrow \Delta(T)$ as:
$$f_i(s) = f(s_0^i \oplus s) $$
Therefore if $ \tau_1, \hdots, \tau_{k-1}$ are the first $k-1$ tokens generated, then the preferred distribution of agent $i$ over the $k$-th token is given by:
\[p_i = f(s_0^i \oplus s \oplus \tau_1 \oplus \cdots \oplus \tau_{k-1}),\]

A key advantage of simulating LLM agents with different prompts is the ability to use a single LLM, making multiple queries with different prompts instead of serving multiple LLMs concurrently.
Because of their large sizes, serving multiple LLMs can be very costly and practically challenging. As one of the key strengths of LLMs, the flexibility to accomplish various tasks with properly designed prompts sheds light to the possibility of training one universal LLM that can, for example, generate different ads according to agent-specific prompts. That is, the universal advertising LLM, plus an advertiser-specific prompt, behaves like an advertiser-specific LLM through the online in-context few-shot learning.

\subsection{Setups}
We illustrate our method with a co-marketing
example as well as a competing brands example. In the setup of co-marketing, the two agents would like to advertise for their brands, ``Alpha Airlines'' and ``Beta Resort'' respectively, regarding a shared topic ``Hawaii.''
We intentionally choose fictitious brands in order to avoid the model directly retrieving any existing ads. We use the brand names ``Alpha'' and ``Beta'' that do not have strong meanings on their own to minimize any potential hallucination, as we are using a common purposed LLM that is not optimized for our task. Each agent is given the following prompt:

\begin{quote}
  ``You are an expert of writing texts that naturally combines two ads together. Your choice of words and sentences is full of artistic flair.

  Write a one-sentence ad for \rule{2cm}{0.15mm}.''
\end{quote}

Agent $A$ uses {\em ``a flight to Hawaii using [\textbf{Alpha Airlines}]''} to fill the blank, while agent $B$ uses {\em ``a vacation in Hawaii at the [\textbf{Beta Resort}]''}. The first two sentences in the prompt aim to improve the quality of the ad generation through {\em assigning roles} (see, for example, \citep{wu2023large}). A natural question is whether the proposed method can adjust the combining strategy according to the context.  Since in both the linear aggregation rule $q_{\mathsf{KL}}$ and the log-linear aggregation rule $\bar q_{\mathsf{KL}}$, there is only one degree of freedom, we parameterize the response by  $\lambda = \nicefrac{b_1}{(b_1 + b_2)}$.

The demonstration for competing brands  \hybridtext{is shown in Appendix~\ref{sec:demo-compete}. In particular, the aggregated results nicely avoid advertising for competing brands together.}{uses the same setup, except that the two brands are ``Beta Resort'' and ``Gamma Hotel''. We seek to understand whether competition will lead LLMs to avoid advertising for both advertisers in the same creative.
}

\hybridtext{}{
\input{www/demo-table}
}

\hybridtext{}{
\input{www/demo-table-compete}
}

\subsection{Results}

The results for the co-marketing example are listed in Table~\ref{tab:demo-short}, where from top to bottom, the value of $\lambda$ decreases from $1$ to $0$. As we can see for both aggregation functions, the generated texts roughly follow the pattern of ``only Alpha Airlines'' $\rightarrow$ ``both Alpha Airlines and Beta Resort'' $\rightarrow$ ``only Beta Resort'' when $\lambda$ goes from $1$ to $0$. This is expected, as $\lambda$ going from $1$ to $0$ corresponds to $b_2$ increasing from $0$ to $\infty$ with $b_1$ fixed (or $b_1$ decreasing from $\infty$ to $0$ with $b_2$ fixed). The thresholds of pattern changes are $0.75$ and $0.4$ for the linear aggregation, and $0.5$ and $0.45$ for the log-linear aggregation.
We emphasize that the example is generated with a general purposed LLM, and it is reasonable to believe that the performance can be improved with proper fine-tuning for the specific task at hand.

\hybridtext{}{
\input{www/demo-compete}

}

%% file: www/demo-table.tex
\begin{table*}[ht]
  \centering
  \small
  \begin{tabular}{c|p{0.42\textwidth}|p{0.42\textwidth}}
    \hline \hline
    $\lambda$ & Linear aggregation function $q_\mathsf{KL}$ & Log-linear aggregation function $\bar q_\mathsf{KL}$ \\
    \hline \hline
    $1$ & \multicolumn{2}{p{0.84\textwidth}}{\textbf{Alpha Airlines}: Your ticket to paradise.} \\
    \hline
    $0.75$ & \textbf{Alpha Airlines}: Fly to Hawaii and experience the beauty of the islands with aloha. \newline \textbf{Beta Resorts}: Stay at our resorts and enjoy the best of Hawaii. & Feel the magic of Hawaii with a flight on \textbf{Alpha Airlines}, now offering 20\% off all flights when you book with us today! \\
    \hline
    $0.6$ & \textbf{Alpha Airlines} flies you to Hawaii, where you can enjoy a week-long stay at the \textbf{Beta resort} for just \$1000. & Experience the magic of Hawaii with a flight on \textbf{Alpha Airlines}, now offering 20\% off all flights when you book with us today!  \\
    \hline
    $0.55$ & \textbf{Alpha Airlines} flies you to Hawaii, where you can enjoy a week-long stay at the \textbf{Beta resort}. & Escape to the tropical paradise of Hawaii with [\textbf{Alpha Airlines}]!  \\
    \hline
    $0.5$ & \textbf{Alpha Airlines} flies you to Hawaii, where you can enjoy a beautiful sunset on the beach. Stay 3 nights and get the 4th free at the \textbf{Beta Resort}. & Experience the magic of Hawaii with a stay at the luxurious [\textbf{Beta Resort}] and a refreshing flight on [\textbf{Alpha Airlines}].  \\
    \hline
    $0.45$ & Fly \textbf{Alpha Airlines} to sunny Hawaii and enjoy the secluded beaches and private lagoons of the Royal Hawaiian \textbf{Beta Resort}. & Experience the magic of Hawaii with a stay at the luxurious [\textbf{Beta Resort}] and a special flight offer from [\textbf{Alpha Airlines}].  \\
    \hline
    $0.4$ & Fly \textbf{Alpha Airlines} to sunny Hawaii and enjoy the first-class treatment that awaits you at \textbf{Beta Resort}, all for one low price. & Experience the magic of Hawaii at the [\textbf{Beta Resort}], where you'll feel like you're in a tropical paradise.  \\
    \hline
    $0.25$ & Experience the magic of Hawaii at the \textbf{Beta Resort}, where the sun shines brighter and the waves crash louder — book your stay today with our exclusive 20\% off discount! & Experience the magic of Hawaii at the \textbf{Beta Resort}, where you'll be pampered like royalty and surrounded by breathtaking beauty.  \\
    \hline
    $0$ & \multicolumn{2}{p{0.84\textwidth}}{Hawaii's \textbf{Beta Resort}: a paradise where the sun shines brighter, the waves sing sweeter, and the sand feels softer.}  \\
    \hline \hline
  \end{tabular}
  \caption{Text generation from two aggregation functions with different $\lambda = b_1 / (b_1 + b_2)$.}
  \label{tab:demo-short}
\end{table*}

%% file: www/demo-table-compete.tex
\begin{table*}[h!]
  \centering
  \small
  \begin{tabular}{c|p{0.42\textwidth}|p{0.42\textwidth}}
    \hline \hline
    $\lambda$ & Linear aggregation function $q_\mathsf{KL}$ & Log-linear aggregation function $\bar q_\mathsf{KL}$ \\
    \hline \hline
    $1$ & \multicolumn{2}{p{0.84\textwidth}}{For those who want a relaxing vacation in Hawaii, the \textbf{Gamma Hotel} is the perfect place to stay.} \\
    \hline
    $0.75$ & Experience the magic of Hawaii at the \textbf{Gamma Hotel}, where the sun shines brighter, the waves crash louder, and the nights are filled with endless possibilities. & Experience the magic of Hawaii at the [\textbf{Gamma Hotel}], where you'll be pampered by our luxurious amenities and stunning views of the Pacific Ocean. \\
    \hline
    $0.6$ & For those who want a truly unforgettable vacation, come to Hawaii and stay at the [\textbf{Gamma Hotel}]. & Plan your dream vacation in Hawaii with our special rates at the luxurious [\textbf{Gamma Hotel}].  \\
    \hline
    $0.55$ & For those who want a truly unforgettable vacation, come to Hawaii and stay at the [\textbf{Gamma Hotel}]. & Plan your dream vacation in Hawaii with our special rates at the luxurious [\textbf{Gamma Hotel}]!  \\
    \hline
    $0.5$ & For those who want a relaxing vacation in Hawaii, the \textit{Gamma Resort} is the perfect place to stay. & \textbf{Gamma Hotel}: a touch of heaven in the heart of Hawaii.  \\
    \hline
    $0.45$ & Hawaii is a beautiful place to vacation, and the \textit{Beta Hotel} is the perfect place to stay. & Plan your dream vacation in Hawaii with our special rates at the luxurious [\textbf{Beta Resort}].  \\
    \hline
    $0.4$ & Escape to the beautiful beaches of Hawaii and stay at the luxurious \textbf{Gamma Hotel}, where you'll be treated like royalty. & Escape to the tropical paradise of Hawaii and stay at the luxurious [\textbf{Beta Resort}], where you'll enjoy stunning views of the Pacific Ocean and easy access to all the island has to offer.  \\
    \hline
    $0.25$ & Escape to the tropical paradise of Hawaii and stay at the luxurious \textbf{Beta Resort}, where you'll be surrounded by lush greenery and stunning ocean views. & Escape to the tropical paradise of Hawaii and stay at the luxurious \textbf{Beta Resort}, where you'll be pampered from the moment you arrive.  \\
    \hline
    $0$ & \multicolumn{2}{p{0.84\textwidth}}{Escape to the tropical paradise of Hawaii and stay at the luxurious \textbf{Beta Resort}, where you'll be pampered from the moment you arrive.}  \\
    \hline \hline
  \end{tabular}
  \caption{Text generation from two aggregation functions with different $\lambda = b_1 / (b_1 + b_2)$ in the competing setup.}
  \label{tab:demo-short-compete}
\end{table*}

%% file: www/demo-compete.tex
The results for the example with competing brands is listed in Table~\ref{tab:demo-short-compete}. As we can see for both aggregation functions, the generated texts roughly follow the pattern of ``only Gamma Hotel'' $\rightarrow$ ``only Beta Resort'' as $\lambda$ goes from $1$ to $0$. The thresholds of pattern changes are $0.5\sim0.45$ for both the linear aggregation and the log-linear aggregation. We note that for $\lambda = 0.45$ and $\lambda = 0.5$, the generated results from the linear aggregation seem to get confused on the brands of ``Gamma Hotel'' and ``Beta Resort'' to generate non-existing ``Gamma Resort'' and ``Beta Hotel''. We expect such phenomenons can be addressed by additional task specific training to enforce the extra requirement on using the exact terms on brands.

%% file: www/conclusion.tex
\section{Conclusion}

In this work, we  put forward a mechanism-design approach to the problem of aggregating the output of several LLMs into one. We have proposed an auction format, the \emph{token auction model}, which operates on a token-by-token basis and allows the LLM agents to influence the output through single-dimensional bids. We explored mechanism design with such LLMs through a robust lens, which makes minimal assumptions about the agents' preferences. Working under this paradigm, we have shown that natural requirements on the incentive properties of the auction mechanism, imply that it should feature a monotone aggregation function. We then showed that under robust preferences, any monotone aggregation function  enables second-price style payments (akin to those in the Vickrey auction or the Generalized-Second Price auction).
We also explored the design of aggregation functions inspired by common loss functions used in LLM training, such as KL-divergence or PPO (in RLHF). As a ``proof of concept'' for our designed mechanism, we demonstrated promising outcomes of our aggregation methods by implementing these aggregation functions in a real-world state-of-the-art LLM using prompt tuning.

%% file: www/appendix.tex
\hybridtext{}{
\section{Omitted Proofs from Section \ref{sec:incentive:mono}}\label{appendix:proof-mono}

Our proofs make use of the following technical lemma, which shows that any complete preference relation on a countable set  can be expressed via a function that maps elements of the countable set into the non-negative rationals.

\begin{lemma}\label{lemma:debreu}
Let $X$ be a countable set, and let $\succeq$ be a complete preference relation on $X$. Then there is a function $f: X \rightarrow \Bids_+$ such that $x \preceq y$ if and only if $f(x) \leq f(y)$.
\end{lemma}

\begin{proof}
Since $X$ is countable, we can arrange it in a sequence $\{x_1, x_2, \ldots\}$. We describe a procedure to define $f(x_i)$ for each $x_i$ in turn. Let $f(x_1) = 1$. Suppose we have defined $f(x_1)$, $f(x_2)$, $\ldots$, $f(x_n)$ in such a way that all order relations are preserved. That is, for all $i,j \leq n$, we have $x_i \preceq x_j$ if and only if $f(x_i) \leq f(x_j)$. We want to define $f(x_{n+1})$. If $x_{n+1} \succeq x_i$ and $x_{n+1} \preceq x_i$ (i.e., $x_{n+1} \sim x_i$) for some $i \leq n$, then we can let $f(x_{n+1}) = f(x_i)$ and we have extended the function to one more element in a way that preserves all order relations until (including) $x_{n+1}$. Otherwise, we can partition the set $\{x_1, \ldots, x_n\}$ into two sets
\[
L = \{x_i:\; i \leq n,\; x_i \prec x_{n+1}\} \quad\text{and}\quad H = \{x_i: \; i \leq n,\; x_i \succ x_{n+1}\}.
\]
In the set $\Bids_+$  of all positive rational numbers, every element  in  $f(L) = \{ f(x_i)\}_{x_i \in L}$   is strictly smaller than  every element in $f(H)$. Choose $q \in \Bids_+$ strictly larger than all the elements of $f(L)$ and strictly smaller than all the elements of $f(H)$. For each $x_i$ with $i \leq n$ the relationship between $x_i$ and $x_{n+1}$ is the same as the relationship between $f(x_i)$ and $q$. Therefore, letting $f(x_{n+1}) = q$ extends the function to one more element, preserving all order relations, as required. The resulting function defined on all of $X$ is thus an isomorphism from $X$ to the image of $f$.
\end{proof}

We remark that the assumption of $X$ being countable in Lemma \ref{lemma:debreu} is crucial as the lemma does not hold for uncountable set $X$. A canonical counter example is  $X = \mathbb{R}^2$ with the lexicographic   ordering: $x \succ x'$ if and only if either $x_1 > x'_1$ or $x_1 = x'_1 \wedge x_2 \geq x'_2$. It is well-known that there is no utility function (continuous or non-continuous) that can induce this complete order \citep{sen2018collective}. This is also why we restrict bids to be non-negative rational numbers, which is certainly practical but also has underlying technical reasons.

  \subsection*{Proof of Lemma~\ref{lemma:full-order}}

\begin{proof}
  We first prove the ``only if'' (``$\Infer$'') direction. That is, suppose $\M = \langle q, z\rangle$ is payment-monotone, then it must imply a total order over $Q(\b_{-i}, \p) = \{ q(b_i, \b_{-i}, \p): b_i \in \Bids_+\}$ for any fixed $\b_{-i}$ and $\p$.

  Fix any $\b_{-i}$ and $\p$. For any $q, q' \in Q(\b_{-i}, \p) = \{ q(b_i, \b_{-i}, \p): b_i \in \Bids_+\}$ such that $q = q(b_i, \b_{-i}, \p)$, $q' = q(b'_i, \b_{-i}, \p)$.  Let $z_i = z_i(b_i, \b_{-i}, \p)$ and $z'_i = z_i(b'_i, \b_{-i}, \p)$ be the corresponding payment given by $\M$. Without loss of generality, suppose that $z_i \geq z'_i$. Then by payment-monotonicity of the mechanism $\M$, we have $q \succeq_i q'$. In other words, $\succeq_i$ establishes a total order over $Q$.

  \vspace{3mm}
\noindent Next we show the ``if'' (``$\Longleftarrow$'') direction. That is, given a total order over $Q(\b_{-i}, \p) = \{ q(b_i, \b_{-i}, \p): b_i \in \Bids_+\}$ for any fixed $\b_{-i}$ and $\p$, we can construct a payment rule $z$ such that $\M = \langle q, z\rangle$ is payment-monotone.

  Fix any $\b_{-i}$ and $\p$. Since $\succeq_i$ establishes a total order over $Q(\b_{-i}, \p)$ and $Q(\b_{-i}, \p)$ is countable, by Lemma~\ref{lemma:debreu} there exists a function $f_{i,\b_{-i},\p}: Q(\b_{-i}, \p) \rightarrow \Bids_+$ such that $\forall q,q' \in Q(\b_{-i}, \p)$,
  \[q \succeq_i q' \iff f_{i,\b_{-i},\p}(q) \geq f_{i,\b_{-i},\p}(q').\]
  Letting $z_i(\b, \p) = f_{i,\b_{-i},\p}(q(\b, \p))$, we obtain a payment-monotone mechanism $\langle q, z \rangle$.
\end{proof}

 \subsection*{Proof of Lemma~\ref{lemma:consistent}}

\begin{proof}

Given the consistency of the distribution aggregation function $q$, we can define a preference relation $\succeq_{i,\p}$ on $\Bids_+$ according to the total orders over $Q(\b_{-i}, \p)$ for all possible $\b_{-i}$ (assumed by the lemma) such that, $\forall b_i, b'_i \in \Bids_+$:
  \begin{itemize}
    \item If there exists $\b_{-i}$ such that $q(b_i, \b_{-i}, \p) \succ_i q(b'_i, \b_{-i}, \p)$, then $b_i$ is preferred to $b'_i$, i.e., $b_i \succ_{i,\p} b'_i$;
    \item If $\forall \b_{-i}$, $q(b_i, \b_{-i}, \p) = q(b'_i, \b_{-i}, \p)$, then $b_i$ and $b'_i$ are indifferent, i.e., $b_i \sim_{i,\p} b'_i$.
  \end{itemize}

  Note that the validity of this construction is guaranteed by the definition of consistency (Definition~\ref{def:consistent-aggregate}):

  \begin{itemize}
  \item For any $b_i, b'_i$ assigned $b_i \succ_{i,\p} b'_i$, there exists some $\b_{-i}$ such that $q(b_i, \b_{-i}, \p) \succ_i q(b'_i, \b_{-i}, \p)$. By the definition of consistency, $\forall \b'_{-i}$, $q(b_i, \b'_{-i}, \p) \succeq_i q(b'_i, \b'_{-i}, \p)$. So we never assign the contradicting order ($b'_i \succ_{i,\p} b_i$).

  \item Meanwhile, for any $b_i, b'_i$ assigned $b_i \sim_{i,\p} b'_i$, we have that $\forall \b_{-i}$, $q(b_i, \b_{-i}, \p) = q(b'_i, \b_{-i}, \p)$, which implies that for all $\b_{-i}$, $q(b_i, \b_{-i}, \p) \not\succ_i q(b'_i, \b_{-i}, \p)$. Therefore, we never assign $b_i \succ_{i,\p} b'_i$.
  \end{itemize}

  By the assumption of the lemma, $\succeq_i$ establishes a total order over $Q(\b_{-i}, \p)$. We argue that this implies that $\succeq_{i,\p}$ is a complete preference relation on $\Bids_+$. Concretely, for every pair $b_i,b'_i \in \Bids_+$ we have $q(b_i), q(b'_i) \in Q(\b_{-i}, \p)$. So either for some $\b_{-i}$, $q(b_i) \succ_i q(b'_i)$ (or $q(b'_i) \succ_i q(b_i)$), or for all $\b_{-i}$, $q(b_i) = q(b'_i)$, due to the lemma's assumption of total order over $Q(\b_{-i}, \p)$. Hence every pair $b_i, b'_i \in \Bids_+$ has an order under $\succeq_{i,\p}$ (i.e., $\succ_{i, \p}$ or $\sim_{i, \p}$), so the preference relation $\succeq_{i, \p}$ is complete.

  By Lemma~\ref{lemma:debreu}, since $\Bids_+$ is countable, there exists a function $f_{i, \p}: \Bids_+ \rightarrow \Bids_+$ that represents the preference relation $\succeq_{i, \p}$, i.e.,
  \begin{gather*}
    \forall b_i, b'_i \in \Bids_+,~b_i \succeq_{i,\p} b'_i \iff f_{i,\p}(b_i) \geq f_{i,\p}(b'_i).
  \end{gather*}
  Hence we can construct $\tilde q$ as:
  \begin{gather*}
    \tilde q(f_{i,\p}(b_i), \b_{-i}, \p) = q(b_i, \b_{-i}, \p).
  \end{gather*}
  Note that for any $b_i, b'_i$ such that $f_{i,\p}(b_i) = f_{i,\p}(b'_i)$, we have $b_i \sim_{i, \p} b'_i \Infer q(b_i, \b_{-i}, \p) = q(b'_i, \b_{-i}, \p)$ for every $\b_{-i}, \p$. Hence, the construction of $\tilde q$ never assigns different values to the same input.

  To complete the construction of $\tilde q(\cdot, \b_{-i}, \p)$, we expand its definition from the image set of $f_{i,\p}$ to the entire bid space while preserving monotonicity. The existence of the expansion is guaranteed by the monotonicity of $\tilde q(\cdot, \b_{-i}, \p)$ on the image set of $f_{i,\p}$ and the fact that the space of aggregations is compact.

  Letting $\pi_i(b) = f_{i,\p}(b)$, we have $q(b_i, \b_{-i}, \p) = \tilde q(\pi_i(b_i), \b_{-i}, \p)$ and $\tilde q$ is a monotone distribution aggregation function for agent $i$. Applying the same argument and the relabeling procedure for every $i$ from $1$ to $n$, completes the proof.
\end{proof}

\subsection*{Proof of Theorem~\ref{thm:monotone}}
\begin{proof}

Applying Lemma~\ref{lemma:full-order}, we know that the payment monotonicity of $\M$ implies a total order over $Q(\b_{-i}, \p)$ for any fixed $\b_{-i}$ and $\p$. Then applying Lemma~\ref{lemma:consistent}, we know that for any $\p$ there exist strategy mappings $\pi_i: \Bids_+ \rightarrow \Bids_+$ for each $i$ such that $q(\b, \p) = \tilde{q}(\pi(\b), \p)$ and $\tilde{q}(\cdot, \p)$ is a monotone aggregation function.

Now, following the same argument in the proof of Lemma~\ref{lemma:consistent}, let us further define $\tilde{z}(\pi(\b),\p) = z(\b,\p)$ for each $\b \in \Bids_+^n$ and $\p \in \Delta(T)^n$. So $\mathcal{M} = \langle q,z \rangle$ is strategically equivalent to $\tilde{\mathcal{M}} = \langle \tilde{q}, \tilde{z} \rangle$ by definition.

It remains to show that the mechanism $\tilde \M = \langle \tilde q, \tilde z\rangle $ is payment-monotone. We know that the original mechanism  $ \M = \langle   q,   z\rangle $ satisfies payment monotonicity, meaning that for each $\p$, $\b_{-i}$, $b_i$, $b'_i$,
\[
z_i(b_i,\b_{-i},\p) \geq z_i(b'_i,\b_{-i},\p)  \Longleftrightarrow q(b_i,\b_{-i},\p) \succeq_i q(b'_i,\b_{-i},\p).
\]
But then, with $\b = (b_i,\b_{-i})$ and $\b' = (b'_i, \b_{-i})$, we also have
\begin{align*}
&\tilde{z}_i(\pi(\b),\p) = z_i(\b,\p) \geq z_i(\b', \p) = \tilde{z}_i(\pi(\b'),\p)\\
\Longleftrightarrow \quad
&\tilde{q}(\pi(\b),\p) = q(\b,\p) \succeq_i q(\b',\p) = \tilde{q}(\pi(\b'),\p),
\end{align*}
so the pair $\tilde{q}$, $\tilde{z}$ satisfies payment monotonicity as needed.
\end{proof}
}

\section{Omitted Proofs from Section~\ref{sec:spa}}\label{app:proofs-spa}

\subsection*{Proof of Lemma~\ref{lemma:monotone-alter}}

\begin{proof}
  We first prove the ``only if'' (``$\Infer$'') direction. Suppose $q$ is a monotone distribution aggregation function. By Definition~\ref{def:mon-agg}, for any agent $i$ and $b'_i \geq b_i \geq 0$, we have
  \[q(b'_i, \b_{-i}, \p) \succeq_i q(b_i, \b_{-i}, \p) \succeq_i q(0, \b_{-i}, \p).\]
  For any undersampled token $t \in \Tunder$, because $q_t(0, \b_{-i}, \p) \leq (p_i)_t$, then by Definition~\ref{def:obvious-prefer}, we have
  \[q_t(b_i, \b_{-i}, \p), q_t(b'_i, \b_{-i}, \p) \in [q_t(0, \b_{-i}, \p), (p_i)_t].\]
  Hence by $q(b'_i, \b_{-i}, \p) \succeq_i q(b_i, \b_{-i}, \p)$, we have
  \[q_t(b'_i, \b_{-i}, \p) \geq q_t(b_i, \b_{-i}, \p),\]
  namely, $q_t(b_i, \b_{-i}, \p)$ weakly increases with $b_i$ and never goes above $(p_i)_t$.

  Similarly, we can prove that for any oversampled
  token $t \in \Tover$, $q_t(b_i, \b_{-i}, \p)$ weakly decreases with $b_i$ and never goes below $(p_i)_t$.

  \medskip
  Then we prove the ``if'' (``$\Longleftarrow$'') direction. Consider any $b'_i \geq b_i$. For any undersampled
  token $t \in \Tunder$, as $q_t(b_i, \b_{-i}, \p) \leq (p_i)_t$ weakly increases with $b_i$, we have
  \[q_t(b_i, \b_{-i}, \p) \leq q_t(b'_i, \b_{-i}, \p) \leq (p_i)_t.\]
  Similarly, we have for any oversampled
  token $t \in \Tover$,
  \[q_t(b_i, \b_{-i}, \p) \geq q_t(b'_i, \b_{-i}, \p) \geq (p_i)_t.\]
  Then by Definition~\ref{def:obvious-prefer},
  \[q(b'_i, \b_{-i}, \p) \succeq_i q(b_i, \b_{-i}, \p),\]
  which then implies the monotonicity of $q$.
\end{proof}

\subsection*{Proof of Theorem~\ref{thm:stable}}

\begin{proof}
Given an aggregation $q(b_i)$ we construct a stable sampling procedure $\sigma$. Let $Q_+(b_i)$ and $Q_-(b_i)$ be the probability of sampling tokens from $\Tunder$ and $\Tover$:
$$\textstyle Q_+(b_i) = \sum_{t \in \Tunder} q_t(b_i),
\qquad
Q_-(b_i) = \sum_{t \in \Tover} q_t(b_i).$$
Both are monotone as $q$ is monotone (by Lemma~\ref{lemma:monotone-alter}).
Reparameterize functions $q_t(b_i)$ in the range $I = [Q_+(0), Q_+(\infty)]$ by defining:\footnote{Here $Q_+^{-1}(x)$ refers to a generalized inverse (or the quantile function \citep{bulow1989simple}) so that it is properly defined even when $Q_+(x)$ is discontinuous.}
$$\hat q_t(x) = q_t(Q_+^{-1}(x)), \forall t \in T, x \in I.$$

Since $\hat q_t(x)$ is monotone, by Lebesgue's Differentiation Theorem, it is differentiable almost everywhere on $I$. We observe that:
\[\textstyle \sum_{t \in \Tunder} \hat q_t(x) = Q_+(Q_+^{-1}(x)) = x, ~~
\sum_{t \in \Tover} \hat q_t(x) = Q_-(Q_+^{-1}(x)) = 1-x.\]
Then $\hat{q}'_t(x)$ for $t \in \Tunder$ forms a probability distribution over $\Tunder$. Similarly $-\hat{q}'_t(x)$ for $t \in \Tover$ forms a probability distribution over $\Tover$. Define $\kappa^+(x), \kappa^-(x) \in \Delta(T)$ as $\kappa^+_t(x) = \hat{q}'_t(x)$ for $t \in \Tunder$, and zero otherwise and $\kappa^-_t(x) = -\hat{q}'_t(x)$ for $t \in \Tover$ and zero otherwise.

We also define the vector $q^+, q^- \in \Delta(T)$ such that $q^+_t = q_t(0)/Q_+(0)$ for $t \in \Tunder$ and zero for $t \in \Tover$. Similarly: $q^-_t = q_t(\infty) / Q_+(\infty)$ for $t \in T^-$ and zero otherwise.

Finally, we define a deterministic function $$\Sampler : \Delta(T) \times [0,1] \rightarrow T$$ that takes a probability vector $p \in \Delta(T)$ and $r \in [0,1]$ and outputs an index $t \in T$ such that $\sum_{j < t} p_j < r \leq \sum_{j \leq t} p_j$.

Now, we are ready to define the stable sampling procedure. Let $\mathcal{R}$ be the uniform distribution on $[0,1]^2$. Given $r = (r_A, r_B) \sim \mathcal{R}$ we define the output $t=\sigma(b_i, r)$ as follows:
\begin{enumerate}
\item if $r_A \leq Q_+(0)$, $t = \Sampler(q^+, r_B) \in \Tunder$;
\item if $Q_+(0) < r_A \leq Q_+(b_i)$, $t = \Sampler(\kappa^+(r_A), r_B) \in \Tunder$;
\item if $Q_+(b_i) < r_A \leq Q_+(\infty)$, $t = \Sampler(\kappa^-(r_A), r_B) \in \Tover$;
\item if $Q_+(\infty) < r_A$, $t = \Sampler(q^-, r_B) \in \Tover$.
\end{enumerate}

Since $\sigma(b_i, r)$ is deterministic, for any fixed $r$, either the output can not be influenced by the bid ($r_A < Q_+(0)$ or $r_A > Q_+(\infty)$) or it can only cause the output to shift from an oversampled token $\Sampler(\kappa^-(r_A), r_B)$ to an undersampled token $\Sampler(\kappa^+(r_A), r_B)$.

We now argue that the tokens are sampled with the correct probabilities. For $t \in \Tunder$, the total probability of getting sampled is:
\begin{gather*}\textstyle \int_0^{Q_+(0)} q^+_t \mathrm{d}r_A + \int_{Q_+(0)}^{Q_+(b)} \hat{q}'_t(r_A) \mathrm{d}r_A = q_t(0) + \hat{q}_t(Q_+(b)) - \hat{q}_t(Q_+(0)) \\
 = q_t(0) + q_t(b) - q_t(0) = q_t(b). \end{gather*}
Similarly for tokens in $\Tover$, the probability of being sampled is:
\begin{gather*} \textstyle \int_{Q_+(b)}^{Q_+(\infty)} -\hat{q}'_t(r_A) \mathrm{d}r_A + \int_{Q_+(\infty)}^1 q^-_t \mathrm{d}r_A  \\ = \hat{q}_t(Q_+(b)) - \hat{q}_t(Q_+(\infty))  + q_t(\infty)
 = q_t(b) - q_t(\infty) + q_t(\infty) = q_t(b). \end{gather*}
This completes the proof.
\end{proof}

\section{Universally Stable Sampling}\label{app:universally-stable}

\begin{example}[Counterexample $4$-token]\label{exmpl:4-token}
  Consider two agents $\{1, 2\}$ and $4$ tokens $\{t_1, t_2, t_3, t_4\}$. Assume that both agents have the same preferred distribution $p_1 = p_2 = (0,0,.5, .5)$ and the allocation function is such that if both agents bid zero the allocation is $(.5, .5, 0, 0)$. Hence both both agents have the same set of favored tokens $\Tunder = \{t_3, t_4\}$ and less favored tokens $\Tover = \{t_1, t_2\}$. The aggregation function $q(b_1, b_2, p_1, p_2)$ is given by the following table:
  \begin{center}
     \begin{tabular}{|c|c|c|}
       \hline
                 & $b_1 = 0$ & $b_1 = 1$ \\
       \hline
       $b_2 = 0$ & $q_{00} = (.5, .5, 0, 0)$& $q_{10} = (0, .5, .5, 0)$  \\
       \hline
       $b_2 = 1$ & $q_{01} = (.5, 0, .5, 0)$ & $q_{11} =  (0, 0, .5, .5)$  \\
       \hline
     \end{tabular}
  \end{center}
  One can verify that the aggregation function is monotone: When either of the agents increase the bid from $0$ to $1$, exactly $1/2$ of the probability mass moves from $\Tover = \{t_1,t_2\}$ to $\Tunder = \{t_3,t_4\}$.

  Now we show that there does not exist a universally stable sampling algorithm that implements this aggregation function. Suppose there exist one, $\sigma$, let $r_A = \{r| \sigma(q_{00}, r) = t_1\}$ and $r_B = \{r | \sigma(q_{00}, r) = t_2\}$. Because $\sigma$ is stable for bidder $1$, when bidder $1$ increases bid, the probability mass only transfers from $\Tover$ to $\Tunder$.

  In this case, when the bid profile $(b_1, b_2)$ moves from $(0, 0)$ to $(1, 0)$, we must have
  \[\sigma(q_{10}, r) = \left\{\begin{array}{ll}t_3, & r \in r_A \\ t_2, & r \in r_B\end{array}\right..\]
  Further apply the same argument when $(b_1, b_2)$ moves from $(1, 0)$ to $(1, 1)$, we have
  \[\sigma(q_{11}, r) = \left\{\begin{array}{ll}t_3, & r \in r_A \\ t_4, & r \in r_B\end{array}\right..\]

  However, if we consider $(b_1, b_2)$ moves along the path $(0, 0) \rightarrow (0, 1) \rightarrow (1, 1)$, we should have
  \[\sigma(q_{01}, r) = \left\{\begin{array}{ll}t_1, & r \in r_A \\ t_3, & r \in r_B\end{array}\right., \qquad \sigma(q_{11}, r) = \left\{\begin{array}{ll}t_4, & r \in r_A \\ t_3, & r \in r_B\end{array}\right..\]
  We end up with a contradiction on the value of $\sigma(q_{11}, r)$ while moving from $(0, 0)$ to $(1, 1)$ along two different paths.
\end{example}
\hybridtext{\input{www/demo-table-compete}}{}

\begin{example}[Counterexample $3$-token]\label{exmpl:3-token}
  Consider two agents $\{1, 2\}$ and $3$ tokens $\{t_1, t_2, t_3\}$, where the agents have different sets of favored (less favored) tokens. In particular, $T_1^+ = \{t_1, t_3\}, T_1^- = \{t_2\}$ and $T_2^+ = \{t_3\}, T_2^- = \{t_1, t_2\}$. The aggregation function is given by the following table:
  \begin{center}
     \begin{tabular}{|c|c|c|}
       \hline
                 & $b_1 = 0$ & $b_1 = 1$ \\
       \hline
       $b_2 = 0$ & $q_{00} = (.5, .5, 0)$& $q_{10} = (.5, 0, .5)$  \\
       \hline
       $b_2 = 1$ & $q_{01} = (0, .5, .5)$ & $q_{11} = (.5, 0, .5)$  \\
       \hline
     \end{tabular}
  \end{center}
  One can verify that the aggregation function is monotone: When $b_1$ increases from $0$ to $1$, exactly $1/2$ of the probability mass moves from $t_2$ to $t_3$ (when $b_2 = 0$) or $t_1$ (when $b_2 = 1$). When $b_2$ increases from $0$ to $1$, either $1/2$ of the probability mass moves from $t_1$ to $t_3$ (when $b_1 = 0$) or no move (when $b_1 = 1$).

  Similarly, suppose that there exists a universally stable sampling algorithm $\sigma$ that implements $q$. Let $r_A = \{r| \sigma(q_{00}, r) = t_1\}$ and $r_B = \{r | \sigma(q_{00}, r) = t_2\}$.

  Following the same argument in Example~\ref{exmpl:4-token}, consider the bid profile $(b_1, b_2)$ moves along the path $(0, 0) \rightarrow (1, 0) \rightarrow (1, 1)$, we must have
  \[\sigma(q_{10}, r) = \sigma(q_{11}, r) = \left\{\begin{array}{ll}t_1, & r \in r_A \\ t_3, & r \in r_B\end{array}\right..\]
  However, consider the bid profile $(b_1, b_2)$ moves along the path $(0, 0) \rightarrow (1, 0) \rightarrow (1, 1)$, we must have
  \[\sigma(q_{01}, r) = \left\{\begin{array}{ll}t_3, & r \in r_A \\ t_2, & r \in r_B\end{array}\right., \qquad \sigma(q_{11}, r) = \left\{\begin{array}{ll}t_3, & r \in r_A \\ t_1, & r \in r_B\end{array}\right..\]
  We end up with a contradiction on the value of $\sigma(q_{11}, r)$.
\end{example}

\section{Omitted Proofs from Section~\ref{sec:aggregation_functions}}\label{app:proofs-agg}

\subsection*{Proof of Proposition~\ref{prop:kl_aggreg}}

\begin{proof}[Proof of Proposition \ref{prop:kl_aggreg}]
  We prove the theorem by showing that the loss $\L_{\KL}^{\bar \mu}$ and the loss in equation \eqref{eq:kl_aggreg_2} differ by a constant and hence have the same minimizer. For  $B = \sum_i b_i$ and a fixed $x$ we will show that:
  \[B \cdot D_\mathsf{KL}(\textstyle \sum_i \frac{b_i}{B} \mu_i(\cdot \vert x) \Vert f^W(x)) - \sum_i b_i D_\mathsf{KL}(f_i(x) \Vert f^W(x)) = \mathsf{const}.\]
  For notation simplicity, we omit the parameters $x$ and $W$ when it is clear from the context and write
  $\sum_i b_i f_i(x) / B = \bar f(x)$.
  Below we treat any term that doesn't depend on $f$ as a constant:
  \begin{align*}
    \sum_i b_i D_\mathsf{KL}(f_i \Vert f) &= \sum_i b_i H(f_i) - \sum_y b_i f_i(y | x) \ln f(y | x)  \\
    &= - B \sum_y \frac{\sum_i b_i f_i(y | x)}B \cdot \ln f(y | x) + \sum_i b_i H(f_i)  \\
    &= -B \cdot H(\bar f, f) + \sum_i b_i  H(f_i) \\
    &= B\cdot D_{\mathsf{KL}}(\bar f \Vert f) - B \cdot H(\bar f) + \sum_i b_i H(f_i)  \\
    &= B \cdot D_{\mathsf{KL}}(\bar f \Vert f) - \mathsf{const}.
  \end{align*}
  To complete the proof, observe that if $f_i$ is the unconstrained minimizer of $\L_{\KL}^{\mu_i}(f)$ we must have $f_i(y \vert x) = \mu_i(y \vert x)$.
  \end{proof}

\subsection*{Proof of Lemma~\ref{lemma:efficient}}

\begin{proof}[Proof of Lemma~\ref{lemma:efficient}]
Let $B = \sum_i b_i$ and consider $\Loss_{\mathsf{KL}}$:

\begin{align*}
  \Loss_{\mathsf{KL}} &= \sum_{i\in[n]} b_i \cdot (H(p_i,q) - H(p_i))  \\
  &= - \sum_{i \in [n]} b_i \cdot \sum_{t \in [T]} p_i(t) \ln q(t) - \sum_{i \in [n]} b_i \cdot H(p_i)   \\
  &= - \sum_{t \in [T]} \ln q(t) \sum_{i \in [n]} b_i \cdot p_i(t) - \sum_{i \in [n]} b_i \cdot H(p_i)   \\
  &= - B \cdot \sum_{t \in [T]} \frac{\sum_{i \in [n]} b_i \cdot p_i(t)}{B} \ln q(t) - \sum_{i \in [n]} b_i \cdot H(p_i)   \\
  &= \sum_{i \in [n]} B \cdot H(q_{\mathsf{KL}}, q) - b_i \cdot H(p_i).
\end{align*}

By Gibbs' inequality, the cross entropy $ H(q_{\mathsf{KL}}, q)$ is minimized if and only if $q = q_{\mathsf{KL}}$. Hence this is also the minimizer of $\Loss_{\mathsf{KL}}$. %
\end{proof}

\subsection*{Proof of Proposition~\ref{prop:rl_aggreg}}

\begin{proof}
For a fixed $x$, $f^*(y \vert x)$ can be obtained by solving:
$$\begin{aligned}
\max_{f(\cdot \vert x)} \text{ } & \sum_y f(y \vert x) \bar r(x,y) - \beta D_\mathsf{KL}(f(x)\Vert f^{\mathsf{SFT}}(x)) \\
\text{s.t. } & \sum_y f(y \vert x) = 1 \text{  and  }  f(y \vert x) \geq 0.
\end{aligned}$$

By the standard KKT conditions, the solution has the form:
$$f^*(y \vert x) = f^{\mathsf{SFT}}(y \vert x) e^{\bar r(x,y) / \beta } C_x.$$
where $C_x$ is a normalization constant to ensure $\sum_y f^*(y \vert x) = 1$. For the same reason, we have that:
$$f_i(y \vert x) = f^{\mathsf{SFT}}(y \vert x) e^{ r_i(x,y) / \beta } C_{i,x}.$$

If $f^\circ$ is the function minimizing problem \eqref{eq:prop_rl}, then we can apply KKT conditions to obtain:
$$ f^\circ(y \vert x) = \exp\left( \frac{1}{B} \sum_i b_i \ln f_i(y \vert x) \right) \cdot C'_x$$
for normalization constants $C'_x$ and $B = \sum_i b_i$. Replacing the formula for $f_i(x)$ from the previous line, we obtain that:
\begin{align*}
f^\circ(y \vert x) & = \exp\left( \frac{1}{B} \sum_i b_i \left( r_i(x,y)/\beta + \ln f^{\mathsf{SFT}}(y \vert x)\right) \right) \cdot C''_x\\
& = \exp\left(   \bar r(x,y)/\beta + \ln f^{\mathsf{SFT}}(y \vert x)\right) \cdot C''_x = f^*(y \vert x).
\end{align*}
This completes the proof.
\end{proof}

\hybridtext{\section{An Example With Competing Ads}\label{sec:demo-compete}

\input{www/demo-compete}}{}